\documentclass[a4paper,USenglish,cleveref,nameinlink,autoref,thm-restate]{lipics-v2021}

\hideLIPIcs
\nolinenumbers

\usepackage{amsmath,amsthm,amsfonts}
\usepackage{mathtools}
\usepackage{hyperref}
\usepackage{microtype}
\usepackage{algorithm}
\usepackage{algpseudocode}
\usepackage[linewidth=1pt]{mdframed}
\usepackage{breqn}
\usepackage[inline]{enumitem}

\newcommand{\cev}[1]{\reflectbox{\ensuremath{\vec{\reflectbox{\ensuremath{#1}}}}}}

\usepackage{graphicx} %
\usepackage{xcolor,soul}
 
\newcommand{\hlg}[1]{\sethlcolor{green}\hl{#1}\sethlcolor{green}}

\newcommand{\hly}[1]{\sethlcolor{yellow}\hl{#1}\sethlcolor{yellow}}

\usepackage{xspace}
\usepackage{todonotes}
\usepackage{complexity}
\newcommand{\myceil}[1]{\ensuremath{\lceil #1 \rceil}}
\newcommand{\myfloor}[1]{\ensuremath{\lfloor #1 \rfloor}}
\usepackage{framed}
\usepackage{tabularx}

\usepackage{tikz}
\usetikzlibrary{calc}

\newlength{\RoundedBoxWidth}
\newsavebox{\GrayRoundedBox}
\newenvironment{GrayBox}[1]%
   {\setlength{\RoundedBoxWidth}{.93\columnwidth}
    \def\boxheading{#1}
    \begin{lrbox}{\GrayRoundedBox}
       \begin{minipage}{\RoundedBoxWidth}}%
   {   \end{minipage}
    \end{lrbox}
    \begin{center}
    \begin{tikzpicture}%
       \node(Text)[draw=black!20,fill=white,rounded corners,inner xsep=2ex,inner ysep=2ex,text width=\RoundedBoxWidth]
             {\usebox{\GrayRoundedBox}};
        \coordinate(x) at (current bounding box.north west);
        \node [draw=white,rectangle,inner sep=3pt,anchor=north west,fill=white]
        at ($(x)+(6pt,.75em)$) {\boxheading};
    \end{tikzpicture}
    \end{center}}

\newenvironment{defproblemx}[2]{\noindent\ignorespaces%
                                \FrameSep=6pt%
                                \parindent=6pt
                \vspace{-3mm}            
                \begin{GrayBox}{#1}%
                \begin{tabular*}{\columnwidth}{!{\extracolsep{\fill}}@{\hspace{.1em}} >{\itshape} p{#2} p{0.84\columnwidth} @{}}%
            }{\\[-1.5ex]
                \end{tabular*}%
                \end{GrayBox}%
                \ignorespacesafterend
                \vspace{-4mm}
            }

\newcommand{\problemQuestion}[3]{%
  \begin{defproblemx}{#1}{1.5cm}
    Input: & #2 \\
    Question: & #3
  \end{defproblemx}
}

\newcommand{\problemQuestionOutput}[3]{%
  \begin{defproblemx}{#1}{1.5cm}
    Input: & #2 \\
    Output: & #3
  \end{defproblemx}
}

\definecolor{defblue}{rgb}{0.121,0.47,0.705}
\DeclareTextFontCommand{\emph}{\color{defblue}\em}

\definecolor{lipicsblue}{rgb}{0.08235294118,0.3098039216,0.537254902}
\definecolor{linkblue}{rgb}{0.098,0.098,0.4392}
\definecolor{ourgreen}{rgb}{0.509,0.745,0.235}
\definecolor{indianred}{rgb}{0.804,0.361,0.361}
\definecolor{indianred1}{rgb}{1,0.416,0.416}
\definecolor{indianred3}{rgb}{0.804,0.333,0.333}
\definecolor{orangered}{rgb}{1,0.271,0}
\definecolor{coral1}{rgb}{1,0.447,0.337}
\definecolor{rosybrown2}{rgb}{0.933,0.231,0.231}
\definecolor{aquamarine4}{rgb}{0.271,0.545,0.455}
\definecolor{chartreuse3}{rgb}{0.4,0.804,0}
\definecolor{mediumpurple3}{rgb}{0.537,0.408,0.804}
\definecolor{mediumvioletred}{rgb}{0.78,0.082, 0.522}

\hypersetup{colorlinks=true,
    linkcolor=orangered,
    anchorcolor=lipicsblue,
    citecolor=lipicsblue,
    filecolor=lipicsblue,
    menucolor=lipicsblue,
    urlcolor=lipicsblue,
    bookmarksopen=true,
    bookmarksopenlevel=2,
    bookmarksnumbered=true,
    plainpages=false,
    }

\newcommand{\bigO}{\mathcal{O}}

\author{Susanna Caroppo}{Roma Tre University, Rome, Italy}{susanna.caroppo@uniroma3.it}{https://orcid.org/0009-0001-4538-8198}{}

\author{Giordano {Da Lozzo}}{Roma Tre University, Rome, Italy}{giordano.dalozzo@uniroma3.it}{https://orcid.org/0000-0003-2396-5174}{}

\author{Giuseppe {Di Battista}}{Roma Tre University, Rome, Italy}{giuseppe.dibattista@uniroma3.it}{https://orcid.org/0000-0003-4224-1550}{}

\authorrunning{Caroppo et al.}

\Copyright{Susanna Caroppo, Giordano {Da Lozzo}, Giuseppe {Di Battista}}

\funding{This research was supported, in part,  by MUR of Italy (PRIN Project no.~2022ME9Z78~-- NextGRAAL and PRIN Project no.~2022TS4Y3N~-- EXPAND).}

\date{}

\title{Quantum Algorithms for One-Sided Crossing Minimization}
\titlerunning{Quantum Algorithms for OSCM}



\begin{document}

\maketitle
\begin{abstract}
We present singly-exponential quantum algorithms for the {\sc One-Sided Crossing Minimization} (OSCM) problem. Given an $n$-vertex bipartite graph $G=(U,V,E\subseteq U \times V)$, a \emph{$2$-level drawing} $(\pi_U,\pi_V)$ of $G$ is described by a linear ordering $\pi_U: U \leftrightarrow \{1,\dots,|U|\}$ of $U$ and linear ordering $\pi_V: V \leftrightarrow \{1,\dots,|V|\}$ of~$V$. For a fixed linear ordering $\pi_U$ of $U$, the OSCM problem seeks to find a linear ordering $\pi_V$ of $V$ that yields a $2$-level drawing $(\pi_U,\pi_V)$ of $G$ with the minimum number of edge crossings. %
We show that OSCM can be viewed as a set problem over $V$ amenable for exact algorithms with a quantum speedup with respect to their classical counterparts.
First, we exploit the quantum dynamic programming framework of Ambainis et al. [{\em Quantum Speedups for Exponential-Time Dynamic Programming Algorithms}. SODA 2019] to devise 
a QRAM-based algorithm that solves  {OSCM} in $\bigO^*(1.728^n)$ time and space.
Second, we use quantum divide and conquer to obtain an algorithm that solves OSCM without using QRAM in $\bigO^*(2^n)$ time and~polynomial~space.
\end{abstract}

\clearpage

\section{Introduction}

We study, from the quantum perspective, the {\sc One-Side Crossing Minimization} (OSCM) problem, one of the most studied problems in Graph Drawing, which is defined below.

\subparagraph{2-Level Drawings.} In a 2-level drawing of a bipartite graph the vertices of the two sets of the bipartition are placed on two horizontal lines and the edges are drawn as straight-line segments. The number of crossings of the drawing is determined by the order of the vertices on the two horizontal lines.
More formally, let $G=(U,V,E)$ be a bipartite graph, where $U$ and $V$ are the two parts of the vertex set of $G$ and $E \subseteq U \times V$ is the edge set of $G$.
In the following, we write $n$, $n_U$, and $n_V$ for $|U \cup V|$, $|U|$, and $|V|$, respectively. Also, for every integer $h$, we use the notation $[h]$ to refer to the set $\{1,\dots,h\}$.
A \emph{2-level drawing} of $G$ is a pair $(\pi_U,\pi_V)$, where $\pi_U : U \leftrightarrow \{1,\dots,|U|\}$ is a linear orderings of $U$, and $\pi_V: V \leftrightarrow \{1,\dots,|V|\}$ is a linear ordering of~$V$. We denote the vertices of $U$ by $u_i$, with $i\in[n_U]$, and the vertices of $V$ by $v_j$, with $j\in[n_V]$.
Two edges $(u_1,v_1)$ and $(u_2,v_2)$ in $E$ \emph{cross} in $(\pi_U,\pi_V)$ if: (i) $u_1 \neq u_2$ and $v_1 \neq v_2$ and (ii) either $\pi_U(u_1)<\pi_U(u_2)$ and $\pi_V(v_2)<\pi_V(v_1)$, or $\pi_U(u_2)<\pi_U(u_1)$ and $\pi_V(v_1)<\pi_V(v_2)$. The number of crossings of a $2$-level drawing $(\pi_U,\pi_V)$ is the number $cr(G,\pi_U,\pi_V)$ of distinct (unordered) pairs of edges that cross.  

Problem OSCM is defined as follows:

\problemQuestionOutput{\sc One-Sided Crossing Minimization (OSCM)}%
{A bipartite graph $G=(U,V,E)$ and a linear ordering $\pi_U : U \leftrightarrow [n_U]$.}%
{A linear ordering $\pi_V : V \leftrightarrow [n_V]$ such that $cr(G,\pi_U,\pi_V)$ is minimum.}

\subparagraph{State of the art.}

The importance of the OSCM problem, which is \NP-complete~\cite{DBLP:journals/algorithmica/EadesW94} even for sparse graphs~\cite{DBLP:conf/gd/MunozUV01}, in Graph Drawing was first put in evidence by Sugiyama in~\cite{DBLP:journals/tsmc/SugiyamaTT81}.

Exact solutions of OSCM have been searched with branch-and-cut techniques, see e.g.~\cite{DBLP:conf/gd/JungerM95,DBLP:conf/isaac/MutzelW98,valls1996branch}, and with FPT algorithms. The parameterized version of the problem, with respect to its natural parameter $k = \min_{\pi_V} cr(G,\pi_U,\pi_V)$,  has been widely investigated. Dujmovic et al.~\cite{DBLP:conf/gd/DujmovicW02,DBLP:journals/algorithmica/DujmovicW04} were the first to show that OSCM can be solved in $f(k)n^{O(1)}$ time, with $f \in O(\psi^k)$, where $\psi \approx 1.6182$ is the golden ratio. Subsequently, Dujmovic and Whitesides~\cite{DBLP:conf/gd/DujmovicFK03,DBLP:journals/jda/DujmovicFK08} improved the running time to $\bigO(1.4656^k+kn^2)$. Fernau et al.~\cite{DBLP:conf/iwoca/FernauFLMPS10}, exploiting a reduction to weighted FAST and the algorithm by Alon et al.~\cite{DBLP:conf/icalp/AlonLS09}, gave a subexponential parameterized algorithm with running time $2^{\bigO(\sqrt{k}\log k)}+n^{\bigO(1)}$. The reduction also gives a PTAS using~\cite{DBLP:conf/stoc/Kenyon-MathieuS07}. Kobayashi and Tamaki~\cite{DBLP:journals/algorithmica/KobayashiT15} gave the current best FPT result with running time $\bigO(k2^{\sqrt{2k}} +n)$.

{\em Quantum Graph Drawing} has recently gained popularity. Caroppo et al.~\cite{CaroppoLB24} applied Grover's search~\cite{DBLP:conf/stoc/Grover96} to several Graph Drawing problems obtaining a quadratic speedup over classical exhaustive search. Fukuzawa et al.~\cite{Fukuzawa2023} studied how to apply quantum techniques for solving systems of linear equations~\cite{DBLP:reference/algo/Harrow16} to Tutte's algorithm for drawing planar $3$-connected graphs~\cite{tutte1963draw}.
Recently, in a paper that pioneered {\em Quantum Dynamic Programming}, several vertex ordering problems related to Graph Drawing have been tackled by Ambainis et al.~\cite{DBLP:conf/soda/AmbainisBIKPV19}.

\subparagraph{Our Results.}

First, we exploit the quantum dynamic programming framework of Ambainis et al.\ to devise an algorithm that solves  {OSCM} in $\bigO^*(1.728^n)$ time and space.
We compare the performance of our algorithm against the algorithm proposed in~\cite{DBLP:journals/algorithmica/KobayashiT15}, based on the value of $k$. We have that the quantum algorithm performs asymptotically better than the FPT algorithm, when~$k \in \Omega(n^2)$.
Second, we use quantum divide and conquer to obtain an algorithm that solves OSCM using $\bigO^*(2^n)$ time and~polynomial~space.
Both our algorithms improve the corresponding classical bounds in either time or space or both.

In our first result, we adopt the QRAM (quantum random access memory) model of computation~\cite{PhysRevLett.100.160501}, which allows (i) accessing quantum memory in superposition and (ii) invoking any $T$-time classical algorithm that uses a (classic) random access memory as a subroutine spending time $\bigO(T)$.
In the second result we do not use the QRAM model of computation since we do not need to explicitly store the results obtained in~partial~computations.

\section{Preliminaries}

We assume familiarity with basic notions in the context of graph drawing~\cite{BattistaETT99}, graph theory~\cite{Diestelbook}, and quantum computation~\cite{DBLP:books/daglib/0046438}. 

\subparagraph{Notation.} For ease of notation, given positive integers $a$ and $b$, we denote $\myceil{\frac{a}{b}}$ as $\frac{a}{b}$ and $\myceil{\log a}$ as $\log a$. 
If $f(n)=\bigO(n^c)$ for some constant $c$, we will write $f(n)=poly(n)$. In case $f(n)=d^n poly(n)$ for some constant $d$, we use the notation $f(n)=\bigO^*(d^n)$ (see, e.g.,~\cite{DBLP:journals/dam/Woeginger08}). 

\subparagraph{Quantum Tools.} The QRAM model of computation enables us to use quantum search primitives that involve condition checking on data stored in random access memory. Specifically, the QRAM may be used by an oracle to check conditions based on the data stored in memory, marking the superposition states that correspond to feasible or optimal solutions.

We will widely exploit the following.

\begin{theorem}[Quantum Minimum Finding, QMF~\cite{DBLP:journals/corr/quant-ph-9607014}]\label{th:quantum-min-find}
    Let $ f: D \rightarrow C $ be a polynomial-time computable function, whose domain $D$ has size $N$ and whose codomain $C$ is a totally ordered set (such as $\mathbb{N}$) and let $\mathcal F$ be a procedure that computes $f$. There exists a bounded-error quantum algorithm that finds 
    $ x \in D $ such that $f(x)$ is minimized 
    using $\bigO(\sqrt{N})$ applications~of~$\mathcal F$.
\end{theorem}

\section[Quantum Dynamic Programming for OSCM]{Quantum Dynamic Programming for One-Sided Crossing Minimization}

In this section, we first describe the quantum dynamic programming framework of Ambainis et al.~\cite{DBLP:conf/soda/AmbainisBIKPV19}, which is applicable to numerous optimization problems involving sets. Then, we show that OSCM is a set problem over $V$ that falls within this framework. We use this fact to derive a quantum algorithm (\cref{th:quantum-exp-time-space}) exhibiting a speedup over the corresponding classical singly-exponential algorithm ~(\cref{th:classical-exp-time-space}) in both time and space complexity. %

\subparagraph{Quantum dynamic programming for set problems.}\label{sse:QDP}
Ambainis et al.~\cite{DBLP:conf/soda/AmbainisBIKPV19} introduced a quantum framework designed to speedup some classical exponential-time and space {\em dynamic programming algorithms}. 
Specifically, the structure of the amenable problems for such a speedup must allow determining the solution for a set $X$ by considering optimal solutions for all partitions $(S,X\setminus S)$ of $X$ with $|S|=k$, for any fixed positive $k$, using polynomial time {\em for each partition}. 
This framework is defined by the following lemma derivable from~\cite{DBLP:conf/soda/AmbainisBIKPV19}. %

\begin{lemma}\label{lem:quantum-dp-lemma}
Let $\mathcal P$ be an optimization problem (say a {\em \em minimization} problem) over a set $X$. Let $|X| = n$ and let $OPT_{\mathcal{P}}(X)$ be the optimal value for $\mathcal P$ over $X$. 
Suppose that there exists a polynomial-time computable function $f_{\mathcal P}: 2^X \times 2^X \rightarrow \mathbb{R}$ such that,
for any $S \subseteq X$, it holds that for any $k \in [|S|-1]$:

\begin{equation}\label{eq:dynamic-recurrence}
OPT_{\mathcal{P}}(S)=\min_{W\subset S, |W|= k} \{OPT_{\mathcal{P}}(W)+OPT_{\mathcal{P}}(S\setminus W)+
f_{\mathcal P}(W,S \setminus W)\}
\end{equation}

\noindent
Then, $OPT_{\mathcal{P}}(X)$ can be computed by a quantum algorithm that uses QRAM in $\bigO^*(1.728^n)$ time and space.
\end{lemma}

\begin{proof}
The algorithm for the proof of the lemma is presented as~\cref{algo:quantum-dp-over-sets}. The main idea of the algorithm is to precompute solutions for smaller subsets using classical dynamic programming and then recombine the results of the precomputation step to obtain the optimal solution for the whole set (recursively) applying QMF (see~\cref{th:quantum-min-find}).
However, to achieve a speedup over the classical dynamic programming algorithm stemming from~\cref{eq:dynamic-recurrence} with $k=|S|-1$, whose time complexity is $\mathcal{O}^*(2^n)$, it must hold that the time complexity of the classical part of the algorithm (L3--L7) and the time complexity of QMF over all subsets must be balanced (L15).

\begin{algorithm}[tb!]
\begin{algorithmic}[1]
\Procedure{$\texttt{QuantumDP}$}{$X$}
\State \textbf{Input}: Set $X$ of size $n$; \textbf{Output}: the value $OPT_{\mathcal{P}}(X)$.
\For{all sets $W \subset X$ such that $|W| \leq (1-\alpha) n/4$} \Comment{in order of increasing size}
\State Compute $OPT_{\mathcal{P}}(W)$ classically via dynamic programming 
\Comment{use~\cref{eq:dynamic-recurrence}\\ \ \ \ \ \ \ \ \ \ \ \ \ \ \ \ \ \ \ \ \ \ \ \ \ \ \ \ \ \ \ \ \ \ \ \ \ \ \ \ \ \ \ \ \ \ \ \ \ \ \ \ \ \ \ \ \ \ \ \ \ \ \ \ \ \ \ \ \ \ \ \ \ \ \ \ \ \ \ \ \ \ \ \ \ \ \ \ \ \ \ \ \ with $k=|W|-1$}
\State Store $OPT_{\mathcal{P}}(W)$ in QRAM
\EndFor
\State \Return $\texttt{OPT}$(X)
\EndProcedure

\Procedure{$\texttt{OPT}$}{$S$}
\State \textbf{Input}: Subset $S \subseteq X$; \textbf{Output}: the value $OPT_{\mathcal{P}}(S)$.
\If{$|S| \leq (1-\alpha)n/4$}
\State \Return value $OPT_{\mathcal{P}}(S)$ stored in QRAM
\Else
\State \Return the result of $\texttt{QMF}$ over all $S \subset X$ %
to find
$$\min_{W\subset S, |W|= \frac{|S|}{2}} \{\texttt{OPT}(W)+\texttt{OPT}(S\setminus W)+ f_{\mathcal P}(W,S \setminus W)\}$$
\EndIf

\EndProcedure
\end{algorithmic}
\caption{Procedure $\texttt{QuantumDP}$ is the algorithm of~\cref{lem:quantum-dp-lemma}.  
Procedure $\texttt{OPT}$ is a recursive procedure invoked by $\texttt{QuantumDP}$.
Procedure $\texttt{QMF}$ performs quantum minimum finding.}
\label{algo:quantum-dp-over-sets}
\end{algorithm}
Specifically,~\cref{algo:quantum-dp-over-sets} works as follows.
Let $\alpha \in (0, 0.5]$ be a parameter. First, it precomputes and stores in QRAM a table containing the solutions for all $W \subset X$ with $|W| \leq (1-\alpha)\frac{n}{4}$ using classic dynamic programming. 
Subsequently, it recursively applies QMF as follows (see Procedure \texttt{OPT}). To obtain the value $OPT_{\mathcal P}(X)$, the {\em first level of recursion} performs QMF over all subsets $S \subset X$ of size $\myfloor{\frac{n}{2}}$ and $\frac{n}{2}$. 
To obtain the value $OPT_{\mathcal P}(S)$ for each of such sets, the {\em second level of recursion} performs QMF over all subsets $Y \subset S$ of size $\myfloor{\frac{n}{4}}$ and $\frac{n}{4}$. 
Similarly, to obtain the value $OPT_{\mathcal P}(Y)$ for each of such sets, the {\em third level of recursion} performs QMF over all subsets $W \subset Y$ of size $\myfloor{\frac{\alpha n}{4}}$ and $\frac{\alpha n}{4}$.
Finally, for any subset of these sizes, the values $OPT_{\mathcal P}(W)$ and $OPT_{\mathcal P}(Y \setminus W)$ can be directly accessed as it is stored in QRAM.

In the following, $H: [0,1] \rightarrow [0,1] $ denotes the {\em binary entropy function}, where $H(p)=-p\log_2(p)-(1-p)\log_2(1-p)$~\cite{DBLP:books/daglib/0013517}.
The overall complexity of~\cref{algo:quantum-dp-over-sets} is as follows:

\begin{itemize}
\item classical pre-processing takes
$\bigO^* \left(\binom{n}{\leq(1-\alpha)\frac{n}{4}}\right) = \bigO^*\left(2^{H\left(\frac{1-\alpha}{4}\right)n}\right)$ time;
\item the quantum part takes 
$ \bigO^*\left(\sqrt{\binom{n}{\frac{n}{2}}\binom{\frac{n}{2}}{\frac{n}{4}}\binom{n}{\frac{\alpha n}{4}}}\right) = \bigO^*\left(2^{\frac{1}{2}\left(1+\frac{1}{2}+\frac{H(\alpha)}{4}\right)n}\right) $ time.
\end{itemize}

The optimal choice for $\alpha$ to balance the classical and quantum parts is approximately $0.055362$, and the resulting space an time complexity of~\cref{algo:quantum-dp-over-sets} are both $\bigO^*(1.728^n)$.
\end{proof}

\subparagraph{Quantum dynamic programming for OSCM.}
In the following, let $(G,\pi_U)$ be an instance of OSCM. %
We start by introducing some notation and definitions.
Let $S$ be a subset of $E$ and let $H=(U,V,S)$ be the subgraph of $G$ whose vertices are those of $G$ and whose edges are those in $S$. For ease of notation, we denote $cr(H,\pi_U,\pi_V)$ simply as $cr_S(\pi_U,\pi_V)$. 
Also, let $\pi_V$ be a linear ordering of the vertices in $V$ and $V_1, V_2 \subseteq V$ be two subsets of the vertices of $V$ such that $V_1 \cap V_2 = \emptyset$. %
We say that $V_1$ \emph{precedes} $V_2$ in $\pi_V$, denoted as $V_1 \prec_{\pi_V} V_2$, if for any $v_1 \in V_1$ and $v_2 \in V_2$, it holds that $\pi_V(v_1) < \pi_V(v_2)$. 
Also, for a any $W \subseteq V$, we denote by $E(W)$ the subset of $E$ defined a follows $E(W):=\{(u_a,v_b): (u_a,v_b) \in E \wedge v_b \in W\}$. 

\smallskip

We will exploit the following useful lemma.

\begin{lemma}\label{lem:separability}
    Let $G=(U,V,E)$ be a bipartite graph and let $\pi_U : U \leftrightarrow [n_U]$ be a linear ordering of the vertices of $U$. 
    Also, let $V_1, V_2 \subseteq V$ be two subsets of the vertices of $V$ such that %
    $V_1 \cap V_2 = \emptyset$. 
    Then, there exists a constant $\gamma(\pi_U,V_1,V_2)$ such that, for 
    every linear ordering $\pi_V : V \leftrightarrow [n_V]$ with $V_1 \prec_{\pi_V} V_2$ we have that: 
    \begin{equation}\label{eq:crossing-separability}
    \hfil
    \gamma(\pi_U,V_1,V_2) = cr_{E(V_1) \cup E(V_2)}(\pi_U,\pi_V) - cr_{E(V_1)}(\pi_U,\pi_V)-cr_{E(V_2)}(\pi_U,\pi_V)
    \end{equation}
\end{lemma}

\begin{proof} 
    First observe that the right side of~\cref{eq:crossing-separability} consists of three terms. The terms $cr_{E(V_1) \cup E(V_2)}(\pi_U,\pi_V)$, $cr_{E(V_1)}(\pi_U,\pi_V)$, and $cr_{E(V_2)}(\pi_U,\pi_V)$ denote
    the number of crossings in $(\pi_U,\pi_V)$ determined by (i) edges in $E(V_1) \cup E(V_2)$, (ii) edges in $E(V_1)$, and (iii) edges in $E(V_2)$, respectively. Therefore, the quantity (i) - (ii) - (iii) represents the number of crossings in $(\pi_U,\pi_V)$ determined by pair of edges such that one edge has an endpoint in $V_1$ and the other edge has an endpoint in $V_2$. 
    Consider two distinct linear orderings $\pi'_V$ and $\pi''_V$ of $V$
    such that $V_1$ precedes $V_2$ in both $\pi'$ and $\pi''$. 
    Consider two edges $e_1 = (u_1,v_1)$, with $v_1\in V_1$, and $e_2 = (u_2,v_2)$, with $v_2 \in V_2$.
    Since $V_1$ precedes $V_2$ in both $\pi'_V$ and $\pi''_V$, we have that $e_1$ crosses $e_2$ in $(\pi_U,\pi'_V)$ and $(\pi_U,\pi''_V)$ only if $\pi_U(u_2) < \pi_U(u_1)$. 
    Therefore, the quantity (i) - (ii) - (iii) determined by $(\pi_U,\pi'_V)$ 
    and $(\pi_U,\pi''_V)$
    does not depend on the specific ordering of the nodes in $V_1$ and $V_2$, but only on the relative position of the sets $V_1$ and $V_2$ within $\pi'_V$ and $\pi''_V$, which is the same in both orders, by hypothesis.
\end{proof}

Observe that, given an ordering $\pi_V$ of $V$ such that $V_1$ precedes $V_2$ in $\pi_V$, the value $\gamma(\pi_U,V_1,V_2)$ represents the number of crossings in a $2$-level drawing $(\pi_U,\pi_V)$ of $G$ determined by pairs of edges, one belonging to $E(V_1)$ and the other belonging to $E(V_2)$.

\smallskip
We are now ready to derive our dynamic programming quantum algorithm for OSCM.
We start by showing that the framework of~\cref{lem:quantum-dp-lemma} can be applied to the optimization problem corresponding to OSCM (i.e., computing the minimum number of crossings over all 2-level drawings $(\pi_U,\pi_V)$ of $G$ with $\pi_U$ fixed). We call this problem {\sc MinOSCM}.

\begin{itemize}
\item First, we argue that {\sc MinOSCM} 
is a set problem over $V$, whose optimal solution respects a recurrence of the same form as~\cref{eq:dynamic-recurrence} of~\cref{lem:quantum-dp-lemma}.
In fact, for a subset $S$ of $V$, let $OPT(S)$ denote the minimum number of crossings in a 2-level drawing $(\pi_U,\pi_S)$ of the graph $G_S = (U,S, E(S))$, where $\pi_S: S \leftrightarrow [|S|]$ is a linear ordering of the vertices of $S$. 
Then, by~\cref{lem:separability}, we can compute $OPT(S)$ by means of the following recurrence for any $k\in[|S|-1]$:

\begin{equation}\label{eq:OSCM-recurrence}
OPT(S)=\min_{W\subset S, |W|= k} \{OPT(W)+OPT(S\setminus W)+
\gamma(\pi_U,W,S \setminus W)\}
\end{equation}

\noindent Clearly, $OPT(V)$ corresponds to the optimal solution for $(G,\pi_U)$. 
Moreover, function $\gamma$ plays the role of function $f_{\mathcal{P}}$ of~\cref{lem:quantum-dp-lemma}. 
\item Second, we have that $\gamma$ can be computed in $poly(n)$ time.
\end{itemize}

Next, we show that~\cref{algo:quantum-dp-over-sets} applied to {\sc MinOSCM} can also be adapted to return an ordering $\pi_V$ of $V$ that yields a drawing with the minimum number of crossings, i.e., a solution for OSCM.
To obtain the optimal ordering $\pi_V$ of $V$, we modify~\cref{algo:quantum-dp-over-sets} as follows. 
We assume that the entries of the dynamic programming table $T$, computed in the preprocessing step of~\cref{algo:quantum-dp-over-sets}, are indexed by subsets of $V$.
When computing the table $T$, for each subset $W\subset V$ with $|W|\leq(1-\alpha)\frac{n}{4}$, together with the value $OPT(W) = \min_{R\subset W, |R|= |W|-1} \{OPT(R)+OPT(W\setminus R)+
f(R,W \setminus R)\}$, we also store a linear ordering $L[W]$ of $W$ such that $cr_{E(W)}(\pi_U,\pi_V) = OPT(W)$. 
Observe that, an optimal ordering of $V$ that achieves $OPT(V)$ is obtained by concatenating an optimal ordering of a subset $S$ of $V$, with $|S| = n/2$, with an optimal ordering of the subset $V \setminus S$, where $S$ is the subset that achieves the minimum value of~\cref{eq:dynamic-recurrence} (where $V$ is the set whose optimal value we seek to compute and $k = n/2$).
Similarly, an optimal ordering for a set $S$ with $|S| = n/2$ is obtained by concatenating an optimal ordering of a subset $Y$ of $S$, with $|Y| = n/4$, with an optimal ordering of the subset $S \setminus Y$, where $Y$ is the subset that achieves the optimal value of~\cref{eq:dynamic-recurrence} (where $S$ is the set whose optimal value we seek to compute and $k = n/4$).
Finally, an optimal ordering for a set $Y$ with $|Y| = n/4$ is obtained by concatenating an optimal ordering of a subset $W$ of $Y$, with $|W| =\alpha\frac{n}{4}$, with an optimal ordering of the subset $Y \setminus W$, where $W$ is the subset that achieves the optimal value of~\cref{eq:dynamic-recurrence} (where $Y$ is the set whose optimal value we seek to compute and $k = \alpha n/4$).
It follows that, an optimal ordering of $V$ that achieves $OPT(V)$  
consists of the concatenation of linear orderings of sets $W_1,W_2,\dots,W_8$, with $|W_i| \leq (1-\alpha)n/4$. We thus modify the procedures $\texttt{QuantumDP}$ and $\texttt{OPT}$  of~\cref{algo:quantum-dp-over-sets} to additionally return such sets. Since the optimal orderings for sets of size bounded by $(1-\alpha)n/4$ are also now stored in $T$. We obtain an optimal ordering
 of $V$ by concatenating the linear orders $L[W_1],L[W_2],\dots,L[W_8]$.

Altogether, we have finally proved the following.

\begin{theorem}\label{th:quantum-exp-time-space}
    There is a bounded-error quantum algorithm that solves OSCM in $O^*(1.728^{n_V})$  time and space.
\end{theorem}

Observe that~\cref{eq:OSCM-recurrence} can also be used to derive an exact classical algorithm for OSCM by processing the subsets of $V$ in order of increasing size. In particular, if $|S| \in O(1)$, then $OPT(S)$ can easily be computed in $poly(n)$ time. Otherwise, by using~\cref{eq:OSCM-recurrence} with $k=|S|-1$, we have that $OPT(S)$ can be computed in $\bigO(|S| poly(n))$ time. Since there exist at most $2^{n_V}$ sets $S\subseteq U$ and since $|S|\leq n_V$, we have the following. 

\begin{theorem}\label{th:classical-exp-time-space}
 There is a classical algorithm that solves OSCM in $\bigO^*(2^{n_V})$ time and space.
\end{theorem}

We remark that~\cref{th:classical-exp-time-space} could also be derived by the dynamic programming framework of Bodlaender et al.~(\cite{DBLP:journals/mst/BodlaenderFKKT12} for linear ordering problems (which exploits a recurrence only involving sets of cardinality $|S|-1$). However, as shown above, in order to exploit~\cref{lem:quantum-dp-lemma}, we needed to introduce the more general recurrence given by~\cref{eq:OSCM-recurrence}. %

Next, we compare our quantum dynamic programming algorithm against the current best FPT result~\cite{DBLP:journals/algorithmica/KobayashiT15} which solves OSCM in $\bigO(k2^{\sqrt{2k}}+n)$ time, where $k$ is the maximum number of crossings allowed in the sought solution.

\begin{corollary}
The algorithm of~\cref{th:quantum-exp-time-space} is asymptotically more time-efficient than the FPT algorithm parameterized by the number $k$ of crossings in~\cite{DBLP:journals/algorithmica/KobayashiT15} when $k \in \Omega(n^2)$.
\end{corollary}

\begin{proof}
By~\cref{th:quantum-exp-time-space}, the computational complexity of our dynamic programming algorithm is $\bigO(poly(n_V)1.728^{n_V})$, where $poly(n)$ is a polynomial function. For ease of computation, we write
$\bigO(poly(n)1.728^{n_V}) = \bigO(2^{\log (poly(n))+n_V\log (1.728)})$.
Hence, upper bounding the time complexity of~\cite{DBLP:journals/algorithmica/KobayashiT15} with only $\bigO(2^{\sqrt{2k}})$ and focusing only on the exponents, we can verify when $\log (poly(n))+n_V\log (1.728)$ is less than $\sqrt{2k}$.
To do that we can upper bound $poly(n)$ with $n^c$ for some constant $c$.  We thus have that $\log (poly(n))+n_V\log (1.728) \leq c \log (n)+n_V\log (1.728) \leq c n +n\log (1.728) \leq \alpha n$, with $\alpha = c+\log (1.728)$. Hence, we have that it is convenient to use our quantum algorithm if $\alpha n < \sqrt{2k}$. That is,  when $k > \frac{\alpha^2}{2} n^2$. 
\end{proof}

\section[Quantum Divide and Conquer for OSCM]{Quantum Divide and Conquer for One-Sided Crossing Minimization}\label{sse:QDC}

Shimizu and Mori~\cite{DBLP:journals/algorithmica/ShimizuM22}  
used divided and conquer to obtain quantum exponential-time polynomial-space algorithms for coloring problems that do not rely on the use of QRAM.
In this section, we first generalize their ideas to obtain a framework designed to speedup, without using QRAM, some classical exponential-time polynomial-space divide and conquer algorithms for set problems. Then, we show that OSCM is a set problem over $V$ that falls within this framework. We use this fact to derive a quantum algorithm ~(\cref{th:quantum-exp-time-pol-space}) that improves the time bounds of the corresponding classical singly-exponential algorithm ~(\cref{th:classical-exp-time-pol-space}), while maintaining polynomial space complexity.

\subparagraph{Quantum divide and conquer for set problems.} 
The quantum divide and conquer framework we present hereafter can be used for set problems with the following features. Consider a problem $\mathcal P$ defined for a set $X$. The
nature of $\mathcal P$ must allow
determining the solution for $X$ by (i)
splitting $X$ into all possible pairs $(S,X \setminus S)$ of subsets of $X$, where $|S| = |X|/2$, (ii) recursively computing the optimal solution for all pairs $(S,X \setminus S)$, 
and (iii) combining the obtained solutions into a solution for $X$ using polynomial time for each of the pairs.
In the remainder, we provide a general quantum framework, defined by the following lemma.

\begin{lemma}\label{lem:quantum-dq-lemma}
Let $\mathcal P$ be an optimization problem (say a {\em \em minimization} problem) over a set $X$. Let $|X| = n$ and let $OPT_{\mathcal{P}}(X)$ be the optimal value for $\mathcal P$ over $X$. 
Suppose that there exists a polynomial-time computable function $f_{\mathcal P}: 2^X \times 2^X\rightarrow \mathbb{R}$ and a constant $c_{\mathcal{P}}$ such that, for any $S\subseteq X$, it holds that:
\begin{enumerate}
    \item If $|S| \leq c_{\mathcal P}$, then  $OPT_{\mathcal P}(S)=f_{\mathcal P} (S,\emptyset)$. 
    \item If $|S| > c_{\mathcal P}$, then  
\end{enumerate}

\begin{equation}\label{eq:divide-recurrence}
OPT_{\mathcal{P}}(S)=\min_{W\subset S, |W|= \frac{|S|}{2}} \{OPT_{\mathcal{P}}(W)+OPT_{\mathcal{P}}(S\setminus W)+
f_{\mathcal P}(W,S \setminus W)\}
\end{equation}

\noindent We have that, $OPT_{\mathcal{P}}(X)$ can be computed by a quantum algorithm without using QRAM in $\bigO^*(2^n)$ time and polynomial space.

\end{lemma}

\begin{proof}

The algorithm for the proof of the lemma is presented as~\cref{alg:quantum-dq-over-sets} and is based on the recurrence in~\cref{eq:divide-recurrence}. The algorithm works recursively as follows. If the input set $X$ is sufficiently small, i.e., $|X|\leq c_{\mathcal{P}}$, then the optimal value for $X$ is computed directly as $f_{\mathcal{P}}(X,\emptyset)$. Otherwise, it uses QMF to find the optimal pair $(S,X \setminus S)$ of subsets of $X$ that determines $OPT_{\mathcal{P}}(X)$ according to~\cref{eq:divide-recurrence}, where the values $OPT_{\mathcal{P}}(S)$ and $OPT_{\mathcal{P}}(X\setminus S)$ have been recursively computed.

\begin{algorithm}[tb]
\caption{The quantum algorithm of~\cref{lem:quantum-dq-lemma}.}\label{alg:RecursiveOSCM}
\begin{algorithmic}[1]
\Procedure{\texttt{QuantumDC}}{$X$}:
\State \textbf{Input}: Set $X$ of size $n$; \textbf{Output}: the value $OPT_{\mathcal{P}}(X)$.
\If{$|S| \leq c_{\mathcal{P}}$ }
\State \Return $f_{\mathcal P}(S,\emptyset)$
\EndIf
\State \Return the result of $\texttt{QMF}$ over all $W \subset S$ with $|W| = \frac{|S|}{2}$ to find 
$$
\min_{W\subset S, |W|= \frac{|S|}{2}} \{\texttt{QuantumDC}(W)+\texttt{QuantumDC}(S\setminus W)+
f_{\mathcal P}(W,S \setminus W)\}
$$
\EndProcedure
\end{algorithmic}
\label{alg:quantum-dq-over-sets}
\end{algorithm}

The running time $Q(k)$ of~\cref{alg:quantum-dq-over-sets} when $|X| = k$ obeys the following recurrence:

$$Q(k)\leq \sqrt{\bigO\left(\binom{k}{{k/2}}\right)}\Big(Q(\myfloor{k/2}) + Q({k/2}) + poly(k)\Big)$$

\noindent Hence, $Q(k)\leq 2^k poly(k)$, and the total running time of~\cref{alg:quantum-dq-over-sets} is bounded by $\bigO^*(2^n)$. %

Finally, the space complexity of 
\cref{alg:quantum-dq-over-sets} (procedure $\texttt{QuantumDC}$) 
can be proved polynomial as follows. A schematic representation of the quantum circuit implementing procedure $\texttt{QuantumDC}$ is shown in~\cref{fig:recursion}.
The execution of $\texttt{QuantumDC}$ determines a rooted binary tree $\mathcal T$ whose nodes are associated with its recursive calls (see~\cref{fig:recursion-tree}). Each such a call corresponds to a circuit in~\cref{fig:recursion}. We denote by $\texttt{QDC(i,j)}$ the circuit, at the $i^{th}$-level of the recursion tree $\mathcal T$ with $i=0,\dots,\log n - 1$, associated with the $j^{th}$-call, with $j \in 0,\dots,{2^i}-1$.
The input to each of such circuits consists of a set of registers defined as follows.
For each $i = 0, 1,\dots,\log n -1$ and $j = 0,1,\dots, 2^i -1 $, there exists a register $A_{i,j}$ with $\frac{n}{2^i}$ qubits. It stores a superposition corresponding to a subset $S_{i,j}$ of $X$ (to be defined later) of size $\frac{n}{2^i}$, which represents all possible ways of splitting the subset into two equal-sized subsets.
Specifically, a status $0$ for $A_{i,j}[k]$ corresponds to assigning the $k^{th}$-element of the subset associated with $A_{i,j}$ to one side of the split, while a status $1$ of $A_{i,j}[k]$ corresponds to assigning the $k^{th}$-element of such a subset to the other side of the split. A suitable quantum circuit allows the qubits to assume only the states where the number of zeros is equal to the number of ones, see e.g.~\cite{CaroppoLB24}.
In~\cref{fig:recursion-tree}, we associate the split defined by the status-$0$ qubits and the split defined by the status-$1$ qubits with the left and right child of a node, respectively. Moreover, in~\cref{fig:recursion-tree}, each edge of $\mathcal T$ is labeled with the registers representing the corresponding splits.

\begin{figure}[tb!]	
		\centering
	\includegraphics[page=8, width=\textwidth]{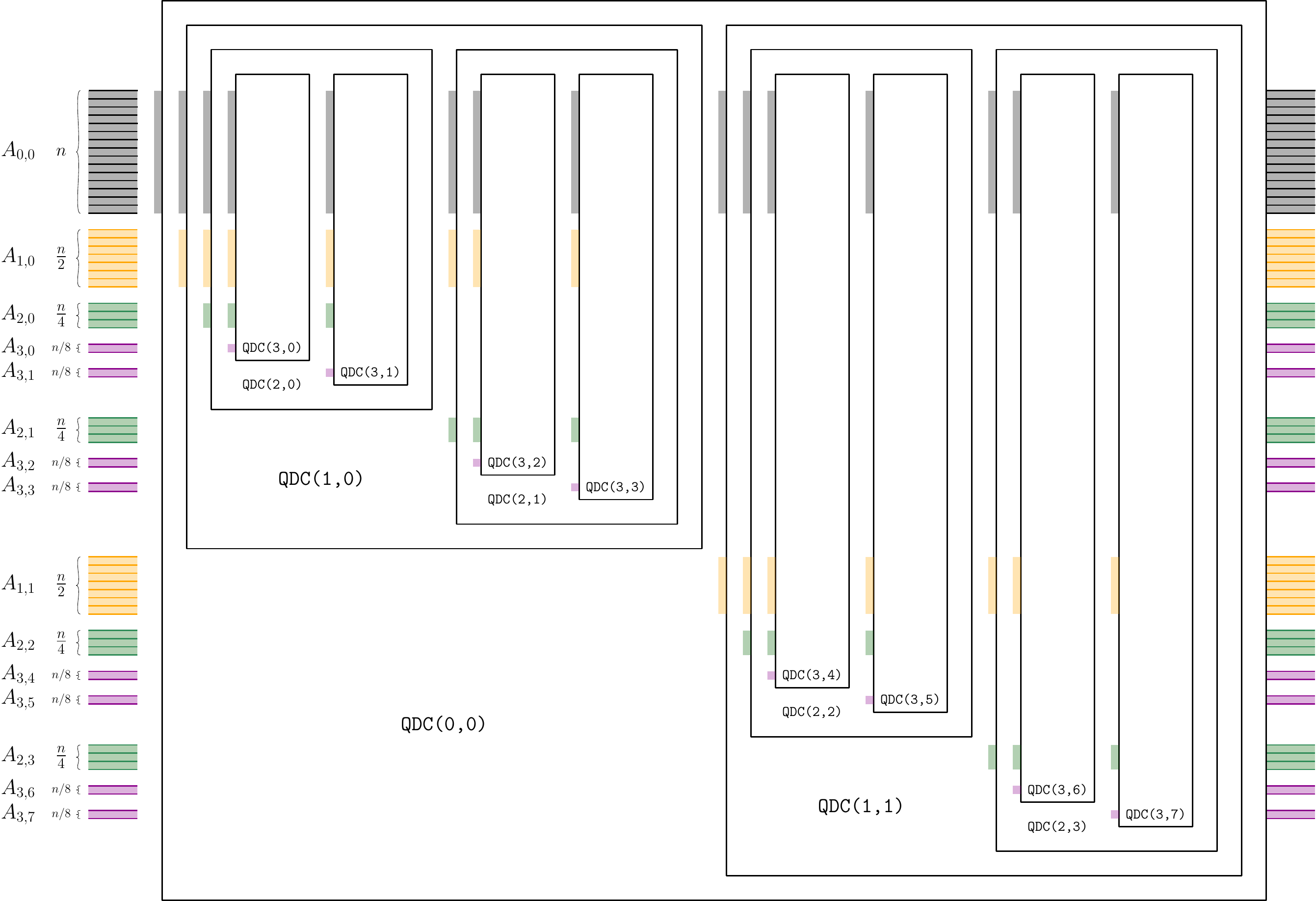}
	\caption{Schematic representation of the circuit realizing~\cref{alg:RecursiveOSCM} for a set $X$ with $n = 16$. The qubits in $L_{i,j}$ in input to the circuit $\texttt{QDC(i,j)}$ are incident to its left boundary. Ancilla qubits are omitted. }
	\label{fig:recursion}
\end{figure}

\begin{figure}[tb!]	
		\centering
	\includegraphics[page=11, width=\textwidth]{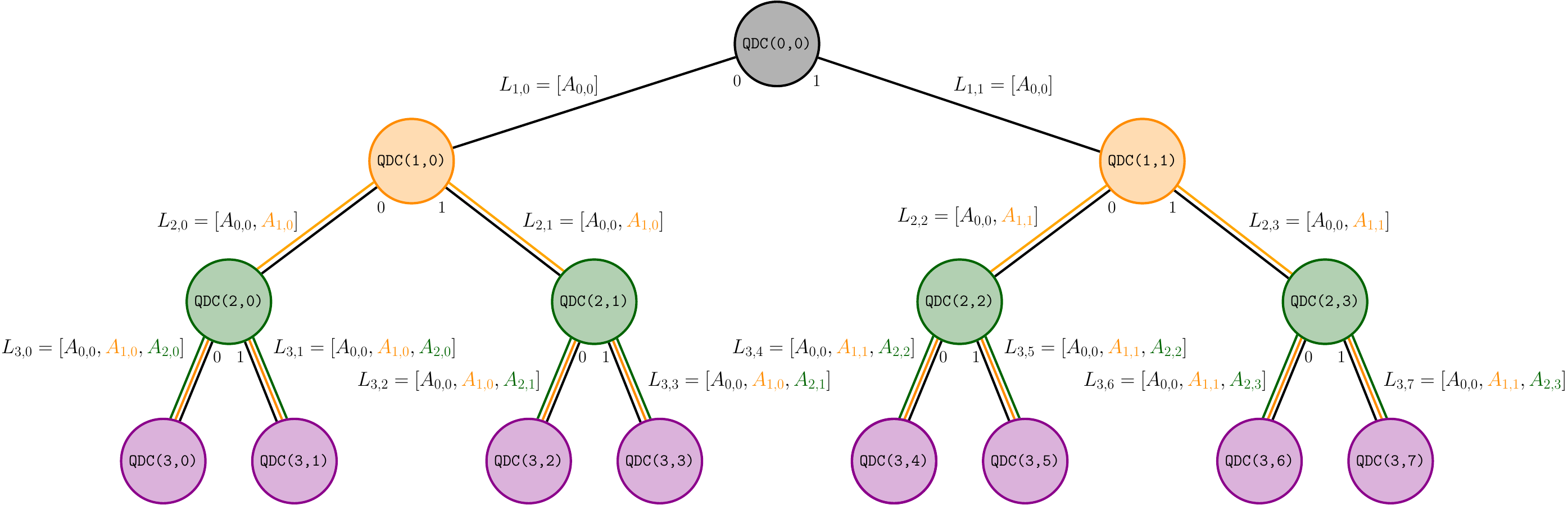}
	\caption{The tree $\mathcal T$ whose nodes are associated with the recursive calls of~\cref{alg:RecursiveOSCM}.}
	\label{fig:recursion-tree}
\end{figure}

The input of $\texttt{QDC(i,j)}$ is a set $L_{i,j}$ of $i+1$ registers of size $n$, $\frac{n}{2}$, $\frac{n}{4}$, $\dots$, $\frac{n}{2^i}$, respectively; see~\cref{fig:recursion}. 
The registers in input to $\texttt{QDC(i,j)}$ can be recursively defined as follows.
The register $A_{i-1,\myfloor{j/2}}$ belongs to $L_{i,j}$ and it is the smallest register in this set. %
Also, if $A_{c,d}$ with $c \geq 1$ belongs to $L_{i,j}$, then $A_{c-1,\myfloor{d/2}}$ also belong to $L_{i,j}$.
In particular, observe that $L_{i,j}$ always contains $A_{0,0}$. 

The circuit $\texttt{QDC(i,j)}$ solves problem $\mathcal P$ on a subset $S_{i,j}$ of $X$ of size $\frac{n}{2^i}$, which is  defined by the states of the registers in $L_{i,j}$.
In particular, the set $S_{i,j}$ can be determined by following the path of $\mathcal T$ connecting $\texttt{QDC(i,j)}$ to the root, and observing that the parity of $j$ determines whether a node in the path is the left or right child of its parent.
For example, consider the circuit $\texttt{QDC(2,2)}$. We show how to determine $S_{2,2}$. Observe that (i) $\texttt{QDC(1,1)}$ is the right child of $\texttt{QDC(0,0)}$, and (ii) $\texttt{QDC(2,2)}$ is the left child of $\texttt{QDC(1,1)}$. 
Also observe that $L_{i,j} = [A_{0,0}, A_{1,1}]$.
To obtain $S_{2,2}$, first by (i) we first consider the subset $S'$ of $X$ corresponding to the qubits in $A_{0,0}$ whose status is $1$, and then by (ii) we obtain $S_{2,2}$ as the subset of $S'$ corresponding to the qubits in $A_{1,1}$ whose status is $0$.

We can finally bound the space complexity of~\cref{alg:quantum-dq-over-sets}, in terms of both classic bits and qubits. Since our algorithm does not rely on external classic memory, we only need to bound the latter. 
We have that the number of circuits $\texttt{QDC(i,j)}$ that compose the circuit implementing the algorithm ~(\cref{fig:recursion}) is the same as the number $\rho$ of nodes of the recursion tree $\mathcal T$. Since $\mathcal T$ is a complete binary tree of height $\log n$, we have that $\rho = 2n - 1$. The number of qubits in $L_{i,j}$ that define the subset of $X$ in input to each circuit $\texttt{QDC(i,j)}$ is at most $\sum^{\log{n}}_{i=0}\frac{n}{2^i} = 2n$. Moreover, the number of ancilla qubits used by each circuit $\texttt{QDC(i,j)}$, omitted in~\cref{fig:recursion}, are polynomial in the number of qubits in the set of registers $L_{i,j}$ in input to $\texttt{QDC(i,j)}$, as they only depend on the size of $S_{i,j}$, which is at most $2n$.  Therefore, the overall space complexity of~\cref{alg:quantum-dq-over-sets} is polynomial. 
\end{proof}

\subparagraph{Quantum divide and conquer for OSCM.}\label{se:qdANDc}
We now describe a quantum divide and conquer algorithm for OSCM. We start by showing that the framework of~\cref{lem:quantum-dq-lemma} can be applied to the optimization problem corresponding to OSCM, which we called {\sc MinOSCM} in~\cref{sse:QDP}. This can be done in a similar fashion as for the~\cref{lem:quantum-dp-lemma}. In particular, the fact that the {\sc MinOSCM} problem is a set problem over $V$ immediately follows from the observation that~\cref{eq:divide-recurrence} is the restriction of~\cref{eq:dynamic-recurrence} to the case in which $k = |W| = \frac{|S|}{2}$. Moreover, recall that $\gamma$ can be computed in $poly(n)$ time.

The execution of \cref{alg:RecursiveOSCM} produces as output a superposition of the registers $A_{i,j}$ such that the state with the highest probability of being returned, if measured, corresponds to an ordering $\pi_V$ of $V$ that yields a drawing with the minimum number of crossings. In the following, we show how to obtain $\pi_V$ from such a state. Recall that, each node $\texttt{QDC(i,j)}$ of $\mathcal T$ is associated with a subset $S_{i,j}$ of $X$.
In particular, the set $S_{0,0}$ for the root node $\texttt{QDC(0,0)}$ coincides with the entire $X$.
To obtain $\pi_V$, we visit $\mathcal T$ in pre-order starting from the root.
When visiting a node of $\mathcal T$, we split the corresponding set $S_{i,j}$ into two subsets $\cev{S_{i+1,2j}}$
and $\vec{S_{i+1,2j+1}}$ based on the value of $A_{i,j}$. In particular, we have that $\cev{S_{i+1,2j}}$ contains the $k^{th}$-vertex in $S_{i,j}$ if $A_{i,j}[k]=0$ and that 
$\vec{S_{i+1,2j+1}}$ contains the $k^{th}$-vertex in $S_{i,j}$ if $A_{i,j}[k]=1$.
We require that, in $\pi_V$, the set $\cev{S_{i+1,2j}}$ precedes the set $\vec{S_{i+1,2j+1}}$.
When the visit reaches the leaves of $\mathcal T$ the left-to-right precedence among vertices in $V$, which defines $\pi_V$, is thus fully specified.

Altogether we have proved the following.

\begin{theorem}\label{th:quantum-exp-time-pol-space}
    There is a bounded-error quantum algorithm that solves OSCM in $\bigO^*(2^{n_V})$ time and polynomial space.
\end{theorem}

Observe that~\cref{eq:divide-recurrence} of~\cref{lem:quantum-dq-lemma} can also be used to derive a classical divide and conquer algorithm for OSCM. 
Clearly, if the input vertex set $X$ is sufficiently small then $OPT(X)$ can be computed in $poly(n)$.
Otherwise, the algorithm considers all the possible splits $(S,X\setminus S)$
of $X$ into two equal-sized subsets, recursively computes the optimal solution for the subinstances induced by each subset and the value $\gamma(\pi_U,S,X \setminus S)$, and then obtains the optimal solution for $X$ by computing the minimum of~\cref{eq:divide-recurrence} over all the considered splits. Clearly, this algorithm can be modified to also return an ordering $\pi_V$ that achieves $OPT(V)$.

The running time of the above algorithm can be estimated as follows. Let $T(k)$ be the running time of the algorithm when $|X| = k$. Clearly, if $k$ is sufficiently small, say smaller than some constant, then $C(k) = poly(k)$. Otherwise, we have that: 

$$T(k)\leq \binom{k}{{k/2}}\Big(T(k/2) + T(\myceil{k/2}) + poly(k)\Big)$$

\noindent Hence, $T(k)\leq 4^k poly(k)$, and the total running time of  the algorithm
is bounded by $\bigO^*(4^n)$.

Therefore, we have the following.

\begin{theorem}\label{th:classical-exp-time-pol-space}
 There is a classical algorithm that solves OSCM in $\bigO^*(4^{n_V})$ time and polynomial space.
\end{theorem}

\section{Conclusions}

In this paper we have presented singly-exponential quantum algorithms for OSCM, exploiting both quantum dynamic programming and quantum divide and conquer. We believe that this research will spark further interest in the design of exact quantum algorithms for hard graph drawing problems. In the following, we highlight two meaningful applications of our results.

\subparagraph{Problem OSSCM.} A generalization of the OSCM problem, called OSSCM and formally defined below,  considers a bipartite graph whose edge set is partitioned into $h$ color classes $E_1,\dots,E_h$, and asks for a 2-level drawing respecting a fixed linear ordering of one of the parts of the vertex set, with the minimum number of \mbox{crossings between edges of the {\em same color}.} %
\problemQuestion{\sc One-Sided Simultaneous Crossing Minimization (OSSCM)}%
{A bipartite graph $G=(U,V,E = E_1\cup E_2\cup \dots \cup E_h)$, a linear ordering $\pi_U : U \leftrightarrow [n_U]$.}%
{An ordering $\pi_V : V \leftrightarrow [n_V]$ such that $\sum_{i=1}^h cr_{E_i}(\pi_U,\pi_V)$ is minimum.}

\medskip
Clearly, OSSCM is a set problem over $V$ whose optimal solution admits a recurrence of the same form as~\cref{eq:dynamic-recurrence,eq:divide-recurrence}. Thus,~\cref{th:classical-exp-time-pol-space,th:classical-exp-time-space} can be extended~to~OSSCM.

\subparagraph{Problem TLCM.} Caroppo et al.~\cite{CaroppoLB24} gave a quantum algorithm to tackle the unconstrained version OSCM, called TLCM and formally defined below, in which both parts of the vertex set are allowed to permute. This algorithm runs in $\bigO^*(2^{\frac{n \log n}{2}})$ time, offering a quadratic speedup over classic exhaustive search. However, the existence of an exact singly-exponential algorithm for TLCM, both classically and quantumly, still appears to be an elusive goal.

\problemQuestion{\sc Two-Level Crossing Minimization (TLCM)}%
{A bipartite graph $G=(U,V,E)$.}%
{Orderings $\pi_U : U \leftrightarrow [n_U]$ and
$\pi_V : V \leftrightarrow [n_V]$ such that $cr(G,\pi_U,\pi_V)$ is minimum.}

\medskip
\cref{th:classical-exp-time-space} allows us to derive the following implication. Consider the smallest between $U$ and $V$, say $U$. Then, we can solve TLCM by performing QMF over all $n_U !$ permutations of $U$ using the quantum algorithm of~\cref{th:classical-exp-time-space} as an oracle. As~$n_U! \leq 2^{n_U \log n_U}$, this immediately yields an algorithm whose running time is $\bigO^*(2^{\frac{n_U\log n_U}{2}}1.728^{n_V})=\bigO^*(2^{\frac{n_U\log n_U}{2}}2^{n_V \log_2 1.728})$. Therefore, as long as $n_V\log_2 1.728 \geq \frac{n_U\log n_U}{2} $, TLCM has a bounded-error quantum algorithm whose running time is $\bigO^*(2^{2(n_V \log_2 1.728)})=\bigO^*(2.986^{n_V})$, and thus singly exponential.

\bibliographystyle{plainurl}
\bibliography{bibliography}

\begin{thebibliography}{10}

\bibitem{DBLP:conf/icalp/AlonLS09}
Noga Alon, Daniel Lokshtanov, and Saket Saurabh.
\newblock Fast {FAST}.
\newblock In Susanne Albers, Alberto Marchetti{-}Spaccamela, Yossi Matias,
  Sotiris~E. Nikoletseas, and Wolfgang Thomas, editors, {\em Automata,
  Languages and Programming, 36th International Colloquium, {ICALP} 2009,
  Rhodes, Greece, July 5-12, 2009, Proceedings, Part {I}}, volume 5555 of {\em
  Lecture Notes in Computer Science}, pages 49--58. Springer, 2009.
\newblock \href {https://doi.org/10.1007/978-3-642-02927-1\_6}
  {\path{doi:10.1007/978-3-642-02927-1\_6}}.

\bibitem{DBLP:conf/soda/AmbainisBIKPV19}
Andris Ambainis, Kaspars Balodis, Janis Iraids, Martins Kokainis, Krisjanis
  Prusis, and Jevgenijs Vihrovs.
\newblock Quantum speedups for exponential-time dynamic programming algorithms.
\newblock In Timothy~M. Chan, editor, {\em Proceedings of the Thirtieth Annual
  {ACM-SIAM} Symposium on Discrete Algorithms, {SODA} 2019, San Diego,
  California, USA, January 6-9, 2019}, pages 1783--1793. {SIAM}, 2019.
\newblock \href {https://doi.org/10.1137/1.9781611975482.107}
  {\path{doi:10.1137/1.9781611975482.107}}.

\bibitem{DBLP:journals/mst/BodlaenderFKKT12}
Hans~L. Bodlaender, Fedor~V. Fomin, Arie M. C.~A. Koster, Dieter Kratsch, and
  Dimitrios~M. Thilikos.
\newblock A note on exact algorithms for vertex ordering problems on graphs.
\newblock {\em Theory Comput. Syst.}, 50(3):420--432, 2012.
\newblock URL: \url{https://doi.org/10.1007/s00224-011-9312-0}, \href
  {https://doi.org/10.1007/S00224-011-9312-0}
  {\path{doi:10.1007/S00224-011-9312-0}}.

\bibitem{CaroppoLB24}
Susanna Caroppo, Giordano {Da Lozzo}, and Giuseppe {Di Battista}.
\newblock Quantum graph drawing.
\newblock In Ryuhei Uehara, Katsuhisa Yamanaka, and Hsu{-}Chun Yen, editors,
  {\em {WALCOM:} Algorithms and Computation - 18th International Conference and
  Workshops on Algorithms and Computation, {WALCOM} 2024, Kanazawa, Japan,
  March 18-20, 2024, Proceedings}, volume 14549 of {\em Lecture Notes in
  Computer Science}, pages 32--46. Springer, 2024.
\newblock \href {https://doi.org/10.1007/978-981-97-0566-5\_4}
  {\path{doi:10.1007/978-981-97-0566-5\_4}}.

\bibitem{BattistaETT99}
Giuseppe {Di Battista}, Peter Eades, Roberto Tamassia, and Ioannis~G. Tollis.
\newblock {\em Graph Drawing: Algorithms for the Visualization of Graphs}.
\newblock Prentice-Hall, 1999.

\bibitem{Diestelbook}
Reinhard Diestel.
\newblock {\em Graph Theory, 4th Edition}, volume 173 of {\em Graduate texts in
  mathematics}.
\newblock Springer, 2012.

\bibitem{DBLP:conf/gd/DujmovicFK03}
Vida Dujmovic, Henning Fernau, and Michael Kaufmann.
\newblock Fixed parameter algorithms for one-sided crossing minimization
  revisited.
\newblock In Giuseppe Liotta, editor, {\em Graph Drawing, 11th International
  Symposium, {GD} 2003, Perugia, Italy, September 21-24, 2003, Revised Papers},
  volume 2912 of {\em Lecture Notes in Computer Science}, pages 332--344.
  Springer, 2003.
\newblock \href {https://doi.org/10.1007/978-3-540-24595-7\_31}
  {\path{doi:10.1007/978-3-540-24595-7\_31}}.

\bibitem{DBLP:journals/jda/DujmovicFK08}
Vida Dujmovic, Henning Fernau, and Michael Kaufmann.
\newblock Fixed parameter algorithms for one-sided crossing minimization
  revisited.
\newblock {\em J. Discrete Algorithms}, 6(2):313--323, 2008.
\newblock URL: \url{https://doi.org/10.1016/j.jda.2006.12.008}, \href
  {https://doi.org/10.1016/J.JDA.2006.12.008}
  {\path{doi:10.1016/J.JDA.2006.12.008}}.

\bibitem{DBLP:conf/gd/DujmovicW02}
Vida Dujmovic and Sue Whitesides.
\newblock An efficient fixed parameter tractable algorithm for 1-sided crossing
  minimization.
\newblock In Stephen~G. Kobourov and Michael~T. Goodrich, editors, {\em Graph
  Drawing, 10th International Symposium, {GD} 2002, Irvine, CA, USA, August
  26-28, 2002, Revised Papers}, volume 2528 of {\em Lecture Notes in Computer
  Science}, pages 118--129. Springer, 2002.
\newblock \href {https://doi.org/10.1007/3-540-36151-0\_12}
  {\path{doi:10.1007/3-540-36151-0\_12}}.

\bibitem{DBLP:journals/algorithmica/DujmovicW04}
Vida Dujmovic and Sue Whitesides.
\newblock An efficient fixed parameter tractable algorithm for 1-sided crossing
  minimization.
\newblock {\em Algorithmica}, 40(1):15--31, 2004.
\newblock URL: \url{https://doi.org/10.1007/s00453-004-1093-2}, \href
  {https://doi.org/10.1007/S00453-004-1093-2}
  {\path{doi:10.1007/S00453-004-1093-2}}.

\bibitem{DBLP:journals/corr/quant-ph-9607014}
Christoph D{\"{u}}rr and Peter H{\o}yer.
\newblock A quantum algorithm for finding the minimum.
\newblock {\em CoRR}, quant-ph/9607014, 1996.
\newblock URL: \url{http://arxiv.org/abs/quant-ph/9607014}.

\bibitem{DBLP:journals/algorithmica/EadesW94}
Peter Eades and Nicholas~C. Wormald.
\newblock Edge crossings in drawings of bipartite graphs.
\newblock {\em Algorithmica}, 11(4):379--403, 1994.
\newblock \href {https://doi.org/10.1007/BF01187020}
  {\path{doi:10.1007/BF01187020}}.

\bibitem{DBLP:conf/iwoca/FernauFLMPS10}
Henning Fernau, Fedor~V. Fomin, Daniel Lokshtanov, Matthias Mnich, Geevarghese
  Philip, and Saket Saurabh.
\newblock Ranking and drawing in subexponential time.
\newblock In Costas~S. Iliopoulos and William~F. Smyth, editors, {\em
  Combinatorial Algorithms - 21st International Workshop, {IWOCA} 2010, London,
  UK, July 26-28, 2010, Revised Selected Papers}, volume 6460 of {\em Lecture
  Notes in Computer Science}, pages 337--348. Springer, 2010.
\newblock \href {https://doi.org/10.1007/978-3-642-19222-7\_34}
  {\path{doi:10.1007/978-3-642-19222-7\_34}}.

\bibitem{Fukuzawa2023}
Shion Fukuzawa, Michael~T. Goodrich, and Sandy Irani.
\newblock Quantum tutte embeddings.
\newblock {\em CoRR}, abs/2307.08851, 2023.
\newblock URL: \url{https://doi.org/10.48550/arXiv.2307.08851}, \href
  {https://arxiv.org/abs/2307.08851} {\path{arXiv:2307.08851}}, \href
  {https://doi.org/10.48550/ARXIV.2307.08851}
  {\path{doi:10.48550/ARXIV.2307.08851}}.

\bibitem{PhysRevLett.100.160501}
Vittorio Giovannetti, Seth Lloyd, and Lorenzo Maccone.
\newblock Quantum random access memory.
\newblock {\em Phys. Rev. Lett.}, 100:160501, Apr 2008.
\newblock URL: \url{https://link.aps.org/doi/10.1103/PhysRevLett.100.160501},
  \href {https://doi.org/10.1103/PhysRevLett.100.160501}
  {\path{doi:10.1103/PhysRevLett.100.160501}}.

\bibitem{DBLP:conf/stoc/Grover96}
Lov~K. Grover.
\newblock A fast quantum mechanical algorithm for database search.
\newblock In Gary~L. Miller, editor, {\em {STOC} 1996}, pages 212--219. {ACM},
  1996.
\newblock \href {https://doi.org/10.1145/237814.237866}
  {\path{doi:10.1145/237814.237866}}.

\bibitem{DBLP:reference/algo/Harrow16}
Aram~W. Harrow.
\newblock Quantum algorithms for systems of linear equations.
\newblock In {\em Encyclopedia of Algorithms}, pages 1680--1683. 2016.
\newblock \href {https://doi.org/10.1007/978-1-4939-2864-4\_771}
  {\path{doi:10.1007/978-1-4939-2864-4\_771}}.

\bibitem{DBLP:conf/gd/JungerM95}
Michael J{\"{u}}nger and Petra Mutzel.
\newblock Exact and heuristic algorithms for 2-layer straightline crossing
  minimization.
\newblock In Franz{-}Josef Brandenburg, editor, {\em Graph Drawing, Symposium
  on Graph Drawing, {GD} '95, Passau, Germany, September 20-22, 1995,
  Proceedings}, volume 1027 of {\em Lecture Notes in Computer Science}, pages
  337--348. Springer, 1995.
\newblock URL: \url{https://doi.org/10.1007/BFb0021817}, \href
  {https://doi.org/10.1007/BFB0021817} {\path{doi:10.1007/BFB0021817}}.

\bibitem{DBLP:conf/stoc/Kenyon-MathieuS07}
Claire Kenyon{-}Mathieu and Warren Schudy.
\newblock How to rank with few errors.
\newblock In David~S. Johnson and Uriel Feige, editors, {\em Proceedings of the
  39th Annual {ACM} Symposium on Theory of Computing, San Diego, California,
  USA, June 11-13, 2007}, pages 95--103. {ACM}, 2007.
\newblock \href {https://doi.org/10.1145/1250790.1250806}
  {\path{doi:10.1145/1250790.1250806}}.

\bibitem{DBLP:journals/algorithmica/KobayashiT15}
Yasuaki Kobayashi and Hisao Tamaki.
\newblock A fast and simple subexponential fixed parameter algorithm for
  one-sided crossing minimization.
\newblock {\em Algorithmica}, 72(3):778--790, 2015.
\newblock URL: \url{https://doi.org/10.1007/s00453-014-9872-x}, \href
  {https://doi.org/10.1007/S00453-014-9872-X}
  {\path{doi:10.1007/S00453-014-9872-X}}.

\bibitem{DBLP:books/daglib/0013517}
David J.~C. MacKay.
\newblock {\em Information theory, inference, and learning algorithms}.
\newblock Cambridge University Press, 2003.

\bibitem{DBLP:conf/gd/MunozUV01}
Xavier Mu{\~{n}}oz, Walter Unger, and Imrich Vrto.
\newblock One sided crossing minimization is np-hard for sparse graphs.
\newblock In Petra Mutzel, Michael J{\"{u}}nger, and Sebastian Leipert,
  editors, {\em Graph Drawing, 9th International Symposium, {GD} 2001 Vienna,
  Austria, September 23-26, 2001, Revised Papers}, volume 2265 of {\em Lecture
  Notes in Computer Science}, pages 115--123. Springer, 2001.
\newblock \href {https://doi.org/10.1007/3-540-45848-4\_10}
  {\path{doi:10.1007/3-540-45848-4\_10}}.

\bibitem{DBLP:conf/isaac/MutzelW98}
Petra Mutzel and Ren{\'{e}} Weiskircher.
\newblock Two-layer planarization in graph drawing.
\newblock In Kyung{-}Yong Chwa and Oscar~H. Ibarra, editors, {\em Algorithms
  and Computation, 9th International Symposium, {ISAAC} '98, Taejon, Korea,
  December 14-16, 1998, Proceedings}, volume 1533 of {\em Lecture Notes in
  Computer Science}, pages 69--78. Springer, 1998.
\newblock \href {https://doi.org/10.1007/3-540-49381-6\_9}
  {\path{doi:10.1007/3-540-49381-6\_9}}.

\bibitem{DBLP:books/daglib/0046438}
Michael~A. Nielsen and Isaac~L. Chuang.
\newblock {\em Quantum Computation and Quantum Information (10th Anniversary
  edition)}.
\newblock Cambridge University Press, 2016.

\bibitem{DBLP:journals/algorithmica/ShimizuM22}
Kazuya Shimizu and Ryuhei Mori.
\newblock Exponential-time quantum algorithms for graph coloring problems.
\newblock {\em Algorithmica}, 84(12):3603--3621, 2022.
\newblock URL: \url{https://doi.org/10.1007/s00453-022-00976-2}, \href
  {https://doi.org/10.1007/S00453-022-00976-2}
  {\path{doi:10.1007/S00453-022-00976-2}}.

\bibitem{DBLP:journals/tsmc/SugiyamaTT81}
Kozo Sugiyama, Shojiro Tagawa, and Mitsuhiko Toda.
\newblock Methods for visual understanding of hierarchical system structures.
\newblock {\em {IEEE} Trans. Syst. Man Cybern.}, 11(2):109--125, 1981.
\newblock \href {https://doi.org/10.1109/TSMC.1981.4308636}
  {\path{doi:10.1109/TSMC.1981.4308636}}.

\bibitem{tutte1963draw}
William~Thomas Tutte.
\newblock How to draw a graph.
\newblock {\em Proceedings of the London Mathematical Society}, 3(1):743--767,
  1963.

\bibitem{valls1996branch}
Vicente Valls, Rafael Mart{\'\i}, and Pilar Lino.
\newblock A branch and bound algorithm for minimizing the number of crossing
  arcs in bipartite graphs.
\newblock {\em European journal of operational research}, 90(2):303--319, 1996.
\newblock \href {https://doi.org/10.1016/0377-2217(95)00356-8}
  {\path{doi:10.1016/0377-2217(95)00356-8}}.

\bibitem{DBLP:journals/dam/Woeginger08}
Gerhard~J. Woeginger.
\newblock Open problems around exact algorithms.
\newblock {\em Discret. Appl. Math.}, 156(3):397--405, 2008.
\newblock URL: \url{https://doi.org/10.1016/j.dam.2007.03.023}, \href
  {https://doi.org/10.1016/J.DAM.2007.03.023}
  {\path{doi:10.1016/J.DAM.2007.03.023}}.

\end{thebibliography}

\clearpage
\appendix

\end{document}